\documentclass[journal,twoside]{IEEEtran}

\usepackage{float}
\usepackage{stackengine}
\usepackage{setspace}
\usepackage{amssymb,amsfonts,latexsym}
\usepackage{graphicx}
\usepackage{subfigure}
\usepackage{psfrag}
\usepackage{algorithmic}
\usepackage{blindtext}
\usepackage{comment}
\usepackage{cite}
\usepackage{url}
\usepackage{times}
\usepackage{bbm}
\usepackage{bbm}
\usepackage{epsfig}
\usepackage{multirow}
\usepackage{longtable}
\usepackage{amsfonts}
\usepackage{amssymb}%
\usepackage{color}
\usepackage{mathrsfs}
\usepackage{bm}
\usepackage{booktabs}
\usepackage{threeparttable}
\usepackage[cmex10]{amsmath}
\usepackage{amsmath}
\usepackage[english]{babel}
\usepackage{mathbbold}
\usepackage{amsthm}
\usepackage{mathtools, nccmath}
\usepackage[linesnumbered,ruled]{algorithm2e}
\makeatother
\usepackage[table]{xcolor}
\usepackage{makecell}
\usepackage{stackengine}

\usepackage{stackengine}

\usepackage{stfloats}

\newtheorem{theorem}{\textbf{Theorem}}

\ifCLASSINFOpdf
\else
\fi

\begin{document}
			\title{Accuracy-Guaranteed Collaborative DNN Inference in Industrial IoT via Deep Reinforcement Learning}
	\author{\small Wen Wu,~\IEEEmembership{\small Member,~IEEE,}
	       Peng~Yang,~\IEEEmembership{\small Member,~IEEE,}
	      Weiting Zhang,~\IEEEmembership{\small Student Member,~IEEE,}\\
	       Conghao Zhou,~\IEEEmembership{\small Student Member,~IEEE,}  and
	        Xuemin~(Sherman)~Shen,~\IEEEmembership{\small Fellow,~IEEE}
	\thanks{W. Wu, C. Zhou, and X. Shen are with the Department of Electrical and Computer Engineering, University of Waterloo, 200 University Avenue West, Waterloo, ON N2L 3G1, Canada (email: \{w77wu, c89zhou, sshen\}@uwaterloo.ca). 	}
	\thanks{P. Yang (Corresponding author) is with the School of Electronic Information and Communications, Huazhong University of Science and Technology, Wuhan, 430074, P.R. China (email: yangpeng@hust.edu.cn). }
	\thanks{W. Zhang is with the School of Electronic and Information Engineering, Beijing Jiaotong University, Beijing 100044, P.R. China  (email: 17111018@bjtu.edu.cn).}}

%
\maketitle

\begin{abstract}

Collaboration among industrial Internet of Things (IoT) devices and edge networks is essential to support
computation-intensive deep neural network (DNN) inference services which require low delay and high accuracy. Sampling rate adaption which dynamically configures the sampling rates of industrial IoT devices according to network conditions, is the key in minimizing the service delay. In this paper, we investigate the collaborative DNN inference problem in industrial IoT networks. To capture the channel variation and task arrival randomness,
we formulate the problem as a constrained Markov decision process (CMDP). Specifically, sampling rate adaption, inference task offloading and edge computing resource allocation are jointly considered to minimize the average service delay while guaranteeing the long-term accuracy requirements of different inference services. Since CMDP cannot be directly solved by general reinforcement learning (RL) algorithms due to the intractable long-term constraints, we first transform the CMDP into an MDP by leveraging the Lyapunov optimization technique. Then, a deep RL-based algorithm is proposed to solve the MDP. To expedite the training process, an optimization subroutine is embedded in the proposed algorithm to directly obtain the optimal edge computing resource allocation. Extensive simulation results are provided to demonstrate that the proposed RL-based  algorithm can significantly reduce the average service delay while preserving long-term inference accuracy with a high probability.

	

	\vspace*{3mm}
\end{abstract}

\section{Introduction}


With the development of advanced neural network techniques and ubiquitous industrial Internet of Things (IoT) devices, deep neural network (DNN) is widely applied in extensive industrial IoT applications, such as facility monitoring and fault diagnosis~\cite{hu2017intelligent}. Industrial IoT devices (e.g., vibration sensors) can sense the industrial operating environment and feed sensing data to a DNN,  and then the DNN processes the sensing data and renders inference results, namely DNN inference. Although DNN inference can achieve high inference accuracy as compared to traditional alternatives {(e.g., decision tree)},  executing DNN inference tasks requires extensive computation resource due to tremendous multiply-and-accumulation operations~\cite{gobieski2019intelligence}. A device-only solution that purely executes DNN inference tasks at resource-constrained industrial IoT devices, becomes intractable due to prohibitive energy consumption and a high service delay. For example, processing an image using AlexNet incurs up to 0.45\;W energy consumption~\cite{chen2016eyeriss}. An edge-only solution which purely offloads large-volume sensing data to resource-rich edge nodes, e.g., access point (AP), suffers from an unpredictable service delay due to time-varying wireless channel~\cite{chekired2018industrial}. Hence, neither a device-only nor an edge-only solution can effectively support low-delay DNN inference services. 

Collaborative inference, which coordinates resource-constrained industrial IoT devices and  the resource-rich AP, becomes a de-facto paradigm to provide low-delay and high-accuracy inference services \cite{li2019edge}. Within the collaborative inference, sensing data from industrial IoT devices can be either processed locally or offloaded to the AP.  At industrial IoT devices, light-weight \emph{compressed} DNNs (i.e., neural networks are compressed without significantly decreasing their performance) are deployed due to constrained on-board computing capability, which saves computing resource at the cost of inference accuracy~\cite{han2015deep, teerapittayanon2016branchynet}. {At the AP, \emph{uncompressed} DNNs are deployed to provide high-accuracy inference services at the cost of network resources.} Through the resource allocation (e.g., task offloading) between industrial IoT devices and the AP, the overall service performance can be enhanced. 

{However, the sampling rate adaption technique that dynamically configures the sampling rates of industrial IoT devices, is seldom considered.} Through dynamically adjusting the sampling rates according to channel conditions and AP's workload, sensing data from industrial IoT devices can be compressed, thereby reducing not only the offloaded data volume, but also task computation workload. In our experiments, we implement AlexNet to conduct bearing fault diagnosis based on the collected bearing vibration signal from dataset~\cite{dataset}.\footnote{{The experiment is conducted on an open-source dataset~\cite{dataset}. This dataset collects the vibration signal of drive end bearings at a sampling rate of 48\;KHz, and  there are 10 types of possible faults.  }} As shown in Fig.~\ref{Fig:sampling_rate}, inference accuracy grows sub-linearly with the sampling rate. For example, when the sampling rate increases from 18\;KHz to 24\;KHz, the accuracy  increases from 95\% to 98.7\%. {Hence, when the channel condition is poor or edge computation workload is heavy, decreasing the sampling rate can reduce the offloaded data volume and requested computation workload, thereby reducing the service delay at the cost of limited inference accuracy. When channel condition is good and edge computation workload is light, increasing the sampling rate can help deliver a high-accuracy service with an acceptable service delay. Hence, sampling rate adaption can effectively reduce the service delay, which should be incorporated in the collaborative DNN inference.}






The sampling rate adaption and resource allocation for collaborative DNN inference are entangled with the following challenges. Firstly, due to time-varying channel conditions and random task arrivals,  sampling rate and resource allocation should be dynamically adjusted to achieve the minimum service delay. Minimizing the long-term service delay requires the stochastic information of network dynamics. Secondly, in addition to minimizing the service delay, the long-term accuracy requirements should be guaranteed for different inference services. The long-term accuracy performance is determined by decisions of sampling rate adaption and resource allocation  over time, and hence the optimal decisions require future network information. To address the above two challenges, a reinforcement learning (RL) technique is leveraged to interact with the unknown environment to capture the network dynamics, and then a Lyapunov optimization technique is utilized within the RL framework to guarantee the long-term accuracy requirements without requiring future network information.


In this paper, we investigate the collaborative DNN inference problem in industrial IoT networks.  \emph{Firstly}, we formulate the problem as a constrained Markov decision process (CMDP)  to account for time-varying channel conditions and random task arrivals. Specifically, sampling rates of industrial IoT devices, task offloading, and edge computation resource allocation are  optimized to minimize the average service delay while guaranteeing the long-term accuracy requirements of multiple services. {\emph{Secondly}, since traditional RL algorithms target at optimizing a long-term reward without considering policy constraints, they cannot be applied to solve CMDP with long-term constraints.} To solve the problem, we transform the CDMP into an MDP via the Lyapunov optimization technique. The core idea is to construct accuracy deficit queues to characterize the satisfaction status of the long-term accuracy constraints, thereby guiding the learning agent to meet the long-term accuracy constraints. \emph{Thirdly}, to solve the MDP, a learning-based algorithm is developed based on the deep deterministic policy gradient (DDPG) algorithm. Within the learning algorithm, to reduce the training complexity, edge computing resource allocation is directly solved via an optimization subroutine based on convex optimization theory, since it only impacts one-shot delay performance according to theoretical analysis. Extensive simulations are conducted to validate the effectiveness of  the proposed algorithm  in reducing the average service delay while preserving the long-term accuracy requirements.

\begin{figure}[t]
	\centering
	\renewcommand{\figurename}{Fig.}
	\includegraphics[width=0.35\textwidth]{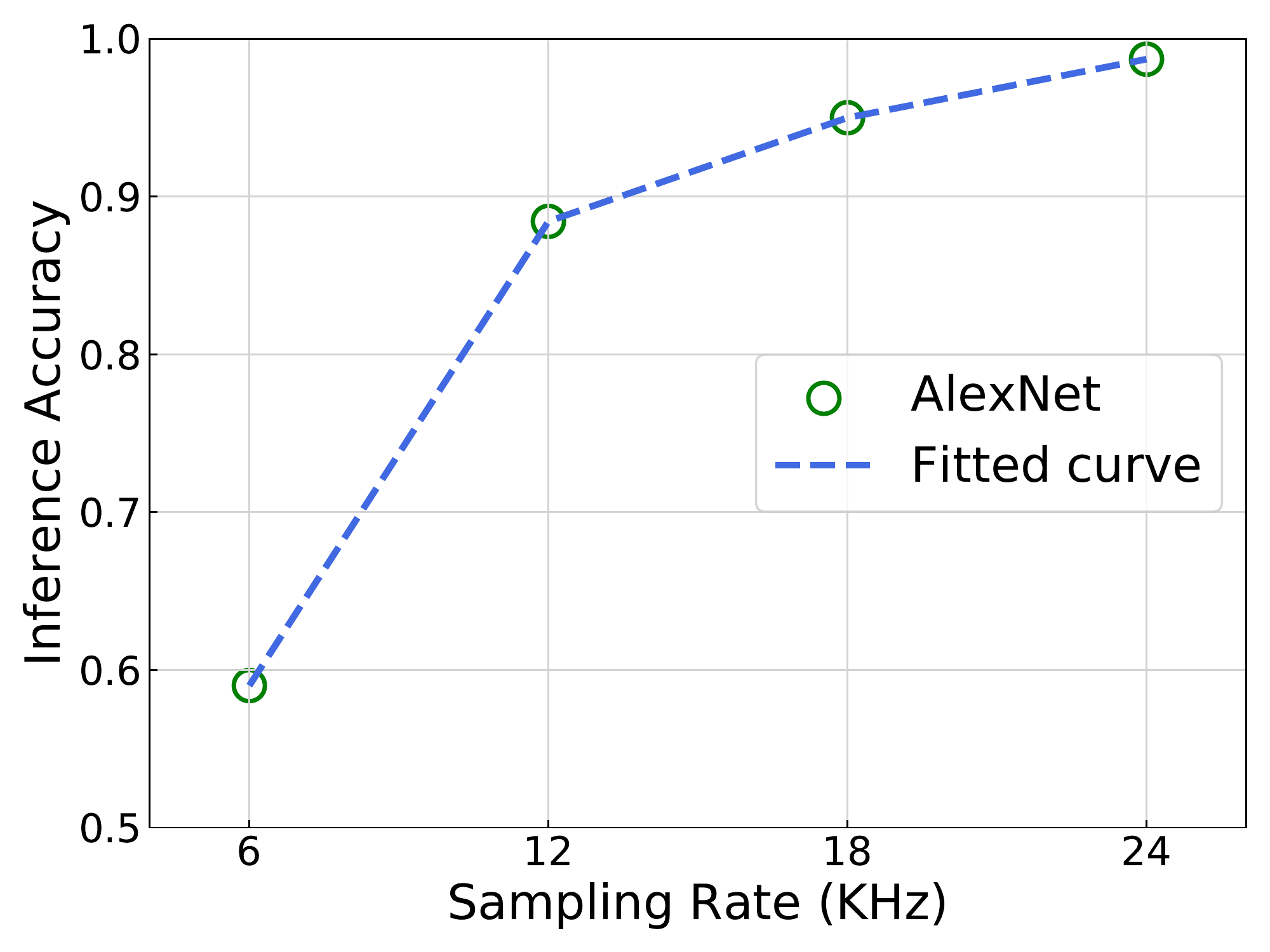}
	\caption{Inference accuracy with respect to sampling rates on the bearing vibration dataset~\cite{dataset}. }
		\vspace{-0.2cm}
	\label{Fig:sampling_rate}
\end{figure}



 Our main contributions in this paper are summarized as follows:
\begin{itemize}
	\item We formulate the collaborative DNN inference problem as a CMDP, in which the objective is to minimize the average service delay while guaranteeing the long-term accuracy constraints;
	\item {We transform the CMDP into an MDP via the Lyapunov optimization technique which constructs accuracy deficit queue to characterize the satisfaction status of the long-term accuracy constraints;}
	\item We propose a deep RL-based algorithm to make the optimal sampling rate adaption and resource allocation decisions. To reduce the training complexity, an optimization subroutine is embedded in the proposed algorithm for the optimal edge computing resource allocation. 
\end{itemize}

The remainder of this paper is organized as follows. Section~\ref{sec: related works} reviews related works. The system model and problem formulation are presented in Section~\ref{sec: system model}. Section~\ref{sec: Proposed solution} proposes a learning-based solution. Simulation results are given in Section~\ref{sec: simulation results}. Finally,  Section~\ref{sec: conclusion} concludes this paper.

\section{Related Work}\label{sec: related works}
DNN inference for resource-constrained industrial IoT devices has garnered much attention recently. A device-only solution aims to facilitate DNN inference services resorting to on-board computing resources. To reduce the computational complexity, DNN compression techniques are applied, such as weight pruning~\cite{han2015deep} and knowledge distillation~\cite{chen2017learning}. Considering the widely-equipped energy-harvesting functionality in IoT devices, Gobieski \emph{et al.} designed a light-weight DNN inference model, which can dynamically compress the model size in order to balance inference accuracy and energy efficiency~\cite{gobieski2019intelligence}. In another line of research, edge-assisted DNN inference solutions can provide high-accuracy inference services by utilizing powerful edge computing servers. To facilitate low-delay and accurate DNN-based video analytics, Yang \emph{et al.} proposed an online video quality and computing resource allocation strategy to maximize video analytic accuracy~\cite{yang2019edge}. Another inspiring work proposed a novel device-edge collaborative inference scheme, in which the DNN model is partitioned and deployed at both the device and the edge, and intermediate results are transferred via wireless links~\cite{li2019edge}. The above works can provide possible resource allocation solutions to enhance DNN inference performance. {Different from existing works, our work takes the sampling rate adaption of industrial IoT devices into account, aiming at providing accuracy-guaranteed inference services in dynamic industrial IoT networks. }

RL algorithms have been widely applied in allocating network resources in wireless networks, {such as service migration in vehicular networks~\cite{wang2019delay}, network slicing in cellular networks~\cite{shen2020ai}, content caching in edge networks~\cite{yang2018content}, and task scheduling in industrial IoT networks~\cite{wang2020online}.} Hence, RL algorithms are considered as plausible solutions to manage network resources for DNN inference services. However, DNN inference services require minimizing the average delay while satisfying the long-term accuracy constraints. {Traditional RL algorithms, e.g., DDPG, can be applied to solve MDPs, in which learning agents seek to optimize a long-term reward without policy constraints, while they cannot deal with constrained long-term optimization problems~\cite{altman1999constrained, liang2018accelerated}.} Our proposed deep RL-based algorithm can address long-term constraints within the RL framework by the modification of reward based on the Lyapunov optimization technique. In addition, an optimization subroutine is embedded in our algorithm to further reduce the training complexity.


%
%
%


\section{System Model and Problem Formulation}\label{sec: system model}
\subsection{Network Model}
\begin{figure}[t]
	\centering
	\renewcommand{\figurename}{Fig.}
	\includegraphics[width=0.48 \textwidth]{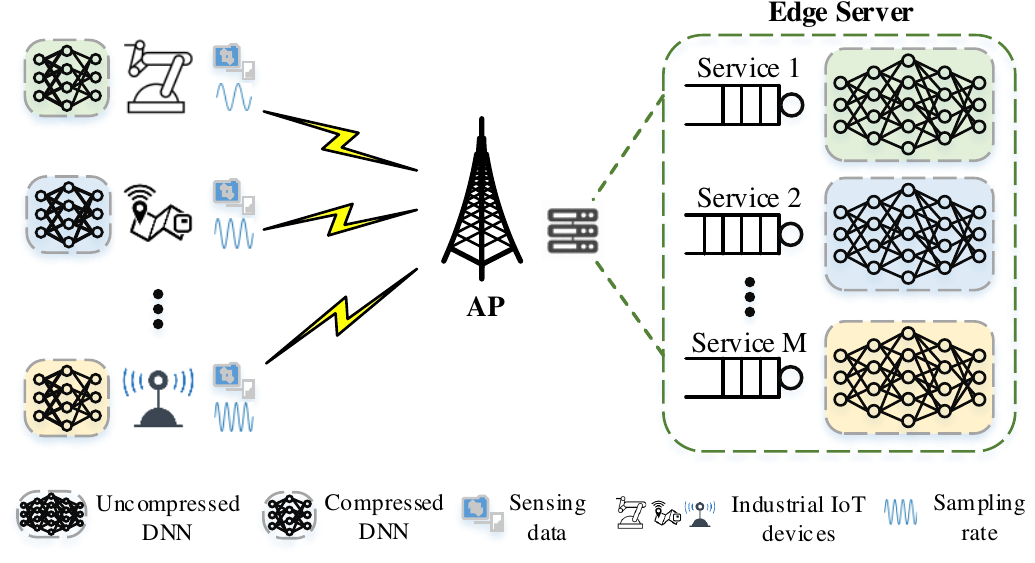}
	\caption{The collaborative DNN inference framework for industrial IoT devices. }
		\vspace{-0.2cm}
	\label{Fig:architecture}
\end{figure}
As shown in Fig. \ref{Fig:architecture}, we consider a wireless network with one AP to serve multiple types of industrial IoT devices. The AP is in charge of collecting network information and resource orchestration within the network. Consider $M$ types of inference services, denoted by a set $\mathcal{M}$, {such as facility fault diagnosis and facility monitoring services.} Taking the facility fault diagnosis service as an example, vibration sensors installed on industrial IoT devices sense the operating conditions at a sampling rate, and feed the sensed vibration signal into a DNN, then the DNN diagnoses the facility fault type. The set of industrial IoT devices subscribed to service $m$ is denoted by $\mathcal{N}_m$, and the set of all industrial IoT devices is denoted by $\mathcal{N}= \cup_{m\in \mathcal{M}}\mathcal{N}_m$. In the collaborative inference framework, two types of DNNs are deployed.  One is a compressed DNN, which is deployed at industrial IoT devices. {The compressed DNN can be implemented via the weight pruning technique, which  prunes less-important weights to reduce computational complexity while maintaining similar inference accuracy~\cite{han2015deep}.}  The other is an uncompressed DNN, which is deployed at the AP. In this way, $M$ types of uncompressed DNNs share the edge computing resource to serve different inference requests. {Important notations are summarized  in Table~\ref{Tab:Variables_and_notations}. }

The collaborative DNN inference framework operates in a time-slotted manner. Let $t$ denote the time index, where $t\in \mathcal{T}=\{1,2,..., T\}$. The detailed procedure is given as follows.
\begin{enumerate}
	\item  Sampling rate selection:  Industrial IoT devices first select their sampling rates according to  channel conditions and computation workloads. The set of candidate sampling rates is denoted by $\mathcal{K}=\{\theta_1, \theta_2,..., \theta_K\}$, where $\theta_K$ denotes the raw sampling rate. {We assume the sampling rate in $\mathcal{K}$ increases linearly with the index, i.e., $\theta_k=k\theta_K/{K}$.}  Let $\mathbf{X}^t$ denote the sampling rate decision matrix in time slot $t$, whose element $x_{n,k}^t=1$ indicates industrial IoT device $n\in \mathcal{N}$ selects the  $k$-th sampling rate. 
	\item Task processing: The sensing data from industrial IoT devices within a time slot is deemed as a computation task, which can be either offloaded to the AP  or executed locally. Let $\mathbf{o}^t \in \mathbb R^{|\mathcal N| \times 1}$ denote the offloading decision vector in time slot $t$, whose element $o_n^t=0$ indicates offloading the computation task from industrial IoT device $n$. Otherwise, $o_n^t=1$ indicates executing the computation task locally.
\end{enumerate}

\begin{table}[t]
	\small
	\centering
	\caption{{Summary of notations}.}
	\label{Tab:Variables_and_notations}
	\vspace{-0.1cm}
	\begin{tabular}{l l l}
		\hline
		\hline
		\textbf{{Notation}} & \textbf{{Description}} \\
		\hline
		{$A_m$} & {Achieved instantaneous accuracy of service $m$}\\
		{$A_m^{th}$} & {Long-term accuracy requirement of service $m$}\\
		{$B^{t}$} & {Local computing queue backlog in time slot $t$}  \\ 
		{$\mathbf{c}$} & {Computing resource allocation decision vector}\\	
		{$D$} & {Service delay}\\
		{$L\left(\cdot\right)$}&{Lyapunov function}\\
		 {$\mathbf{o}$ }& {Task offloading decision vector} \\
		 {$Q^{t}$ }&{Edge computing queue backlog in time slot $t$ }\\

		 {$V$ }& {Parameter to balance delay and accuracy requirement}\\
		 {$\mathbf{X}$ }& {Sampling rate selection decision matrix of all  devices}\\
		 {$Z^{t}$ }& {Accuracy deficit queue backlog in time slot $t$ }\\
		{$\xi_n$}& {Raw task data size of device $n$}\\ 
		{$\eta_m$ }&{Task computation intensity of service $m$}\\
	{$\lambda_n$ }&{Average task arrival rate of device $n$}\\
		{$\zeta\left(\mathbf{x}_n^t\right)$ }& {Task data size of device $n$ in time slot $t$}\\
		{$\Psi^t$ }& {Amount of dropped tasks in computing queues}\\
		\hline
		\hline
	\end{tabular}
	\vspace{-0.2cm}
\end{table}


%


\subsection{Service Delay Model}\label{sec:Service Delay Model}

A computation task can be either processed locally or offloaded to the AP. In what follows, we analyze the service delay in these two cases.
\subsubsection{Executing locally}
 {Let $\lambda_n^t$ denote the task arrival rate of the $n$-th industrial IoT device in time slot $t$, which is assumed to follow  a general random distribution. The raw data size of the generated tasks at the $n$-th device is denoted by $\xi_n^t=\lambda_n^t \nu_m, \forall n\in \mathcal{N}_m$, where $\nu_m$  denotes the raw data size of a task for service $m$.} After  the sampling rate is selected, the data size of the generated task is represented by $\zeta\left(\mathbf{x}_{n}^t\right)=\sum_{k=1}^{K} x_{n, k}^t \xi_n^t k/K$, where $\mathbf{x}_n^t=\{x_{n, k}^t\}_{k\in \mathcal{K}}$ is the sampling rate selection decision vector of the $n$-th device. When the inference task is processed locally by a compressed DNN, the service delay includes the queuing delay in the local computing queue and task processing delay, which is given by 
\begin{equation}
d^t_{n,l}=\frac{o_n^t\eta_{m,c} \left(B_n^t+\zeta\left(\mathbf{x}_{n}^t\right) \right)}{f_n}, \forall n\in \mathcal{N}_m,
\end{equation} 
where $f_n$ is the CPU frequency of the $n$-th industrial IoT device, and $\eta_{m,c}$ denotes the computation intensity of the compressed DNN for the $m$-th service. Here, $B_n^t$ is the backlogged computation tasks (in bits) in the local computing queue, which is updated via
\begin{equation}\label{equ:local_buffer}
B_n^{t+1}=\min\left\{\left[ B_n^t+ o_n^t\zeta\left(\mathbf{x}_n^t\right) -\frac{f_n\tau }{\eta_{m,c}}\right]^+, B_n^{max}\right\}, 
\end{equation}
where $\left[x \right]^+=\max\{x,0\}$, $B_n^{max}$ is the capacity of the local computing queue, and $\tau$ is the duration of a time slot. {Tasks will be dropped if the local computing queue is full. Let  
\begin{equation}
\Psi^t_{b,n}=\max \left\{B_n^t+ o_n^t\zeta\left(\mathbf{x}_n^t\right) -\frac{f_n\tau }{\eta_{m,c}}- B_n^{max}, 0 \right\}
\end{equation}
denote the amount of the dropped tasks in the local computing queue of device $n$. Here, $\Psi^t_{n,b}>0$ indicates that an event of local computing queue overflow occurs at the $n$-th  device, and the corresponding penalty will be incurred to avoid queue overflow.}

%


\subsubsection{Offloading to AP}
When a task is offloaded to the AP, it will be processed by an uncompressed DNN. The service delay consists of task offloading delay, queuing delay in the edge computing queue, and task processing delay, which are analyzed respectively as follows. 
\begin{itemize}
	\item Task offloading delay: For the $n$-th industrial IoT device, the offloading delay is given by
	\begin{equation}
	d_{n, o}^t=\frac{\left(1-o_n^t\right)\zeta\left(\mathbf{x}_n^t\right)}{R_n^t},
	\end{equation}
	where transmission rate between the $n$-th industrial IoT device  and the AP, $R_n^t$, is given by
	\begin{equation}
	R_n^t=\frac{W}{N}\log_2\left(1+\frac{P_TG(H_n^t)}{N_f\sigma^2}\right).
	\end{equation}
	Here, $W$, $P_T$, $G(H_n^t)$, and $N_f$ represent the system bandwidth, transmit power, channel gain, and noise figure, respectively. $\sigma^2={N_oW}/{N}$ denotes the background noise where $N_o$ is  thermal noise spectrum density. Channel gain $G(H_n^t)$ varies in terms of channel state $H_n^t$. Based on extensive real-time measurements,  channel state $H_n^t$ can be modeled with a finite set of channel states $\mathcal{H}$~\cite{lei2016optimal}. The evolution of channel states is characterized by a discrete-time and ergodic Markov chain model, whose transition matrix is $\mathbf{P}\in \mathbb{R}^{|\mathcal{H}|\times |\mathcal{H}|}$. 
	
	\item Task processing delay: The  tasks from all industrial IoT devices subscribed to the $m$-th service are placed in the edge computing queue for the $m$-th service. The amount of aggregated tasks is given by $\sum_{n\in \mathcal{N}_m}\left(1-o_n^t\right)\zeta\left(\mathbf{x}_n^t\right)$.  The computing resource is dynamically allocated among multiple services at the AP according to service task arrivals, which can be realized via containerization techniques, such as Dockers and Kubernetes~\cite{bernstein2014containers}. Let $\mathbf{c}^t\in \mathbb{R}^{M\times 1}$ denote the computing resource allocation decision vector in time slot $t$. Each element $0\leq c_m^t\leq 1$ denotes the portion of computing resource allocated to the $m$-th service. 
	Hence, the processing delay  is given by 
	\begin{equation}
	d^t_{n, p}=\frac{ \eta_{m,u} \left(1-o_n^t\right)\zeta\left(\mathbf{x}_n^t\right)}{c_m^t f_b}, \forall n\in \mathcal{N}_m,
	\end{equation}
	 where $f_b$ is the CPU frequency of the computing server at the AP, and $\eta_{m,u}$ denotes the computation intensity of processing the $m$-th service task by the uncompressed DNN. Note that $\eta_{m,u}>\eta_{m,c}$, since the uncompressed DNN consumes more computing resource.
	
	\item Queuing delay: The queuing delay consists of two components: (i) the time taken to process  backlogged tasks in the edge computing queue, which is given by 
	\begin{equation}
	d^t_{n,q}=\frac{ Q_{m}^t\eta_{m,u}  }{c_m^t f_b}, \forall n\in \mathcal{N}_m.
	\end{equation}
 Here, $Q_{m}^t$ denotes the edge computing queue backlog for the $m$-th service in time slot $t$, 
 which is updated according to 
	\begin{equation}\label{equ:computing_queue_update}
	\begin{split}
		Q_{m}^{t+1}&=\min\left\{ \left[ Q_{m}^t+ a_m^t - \frac{c_m^tf_b\tau}{\eta_{m,u}}\right]^+, Q_m^{max} \right\}.
	\end{split}
	\end{equation}
	Here, $a_m^t=\sum_{n\in \mathcal{N}_m} \left(1-o_n^t\right)\zeta\left(\mathbf{x}_n^t\right) $ and $Q_m^{max}$ denotes the capacity of the $m$-th edge computing queue. Similar to that in local computing queues, tasks will also be dropped if the edge computing queue is full, and the amount of dropped tasks for the $m$-th edge computing queue is given by
	\begin{equation}
	\Psi^t_{q,m}=\max \left\{ Q_{m}^t+ a_m^t - \frac{c_m^tf_b\tau}{\eta_{m,u}}-Q_m^{max}, 0\right\}.
	\end{equation}
	Here, $	\Psi^t_{q,m}>0$ indicates that an event of edge computing queue overflow occurs; and (ii) average waiting time among all newly arrived tasks until the task of industrial IoT device $n$ is processed, which is given by
	\begin{equation}
	d^t_{n,w}=\frac{\eta_{m,u}   \sum_{i\neq n, i \in \mathcal{N}_m }\left(1-o_{i}^t\right) \zeta\left(\mathbf{x}_i^t\right) }{ 2 c_m^t f_b}.
	\end{equation}
	Here, $\sum_{i\neq n, i \in \mathcal{N}_m }\left(1-o_{i}^t\right) \zeta\left(\mathbf{x}_n^t\right)$ denotes the amount of aggregated tasks except the task of industrial IoT device~$n$. 
\end{itemize}


Taking both local execution and offloading into account, the service delay in time slot $t$ is given by
\begin{equation}\label{equ:overall_delay}
\begin{split}
D\left(\mathbf{X}^t, \mathbf{o}^t, \mathbf{c}^t\right)
&=\sum_{n\in \mathcal{N}}\left(d^t_{n,l}+ d^t_{n,o}+d^t_{n, p}+d^t_{n,q}+d^t_{n,w}\right)\\
&+ w_p\left(\sum_{ n\in \mathcal{N}} \mathbbm{1}_{\{\Psi_{b,n}^t>0\}} + \sum_{ m\in \mathcal{M}} \mathbbm{1}_{\{\Psi_{q,m}^t>0\}}\right),
\end{split}
\end{equation}
where $\mathbbm{1}_{\{x\}}=1$ and $w_p>0$ are the indicator function and the positive unit penalty cost for queue overflow, respectively. The first term represents the experienced delay  to complete all tasks in time slot $t$. The second term represents the penalty for potential overflow events in local and edge computing queues.


\subsection{Inference Accuracy Model}\label{sec:Inference Accuracy Model}
{The inference accuracy depends on the sampling rate of a task and the  type of DNN that executes a task. Firstly, we characterize the relationship between the inference accuracy and the sampling rate, which is specified by accuracy function $g(\theta_k), \forall \theta_k \in \mathcal{K}$. Specifically, we implement a DNN inference algorithm, i.e., AlexNet~\cite{krizhevsky2012imagenet}, and apply the AlexNet to diagnose facility fault type based on the collected bearing vibration signal from the dataset~\cite{dataset}, and then measure the accuracy function values with respect to sampling rates, as shown in Fig.~\ref{Fig:sampling_rate}.} 
{Secondly, the relationship between the inference accuracy and the type of DNN  is also characterized via experiments.} Here, $h_{m, c}$ and $h_{m, u}$ represent the inference accuracy of the compressed DNN and the uncompressed DNN for the $m$-th service, respectively. Note that, $h_{m, c}<h_{m, u}$, as an uncompressed DNN achieves higher fault diagnosis accuracy. 

Since the DNN model selection (i.e., task offloading decision) and the sampling rate selection are independent,  inference accuracy is the product of the accuracy value with respect to the selected sampling rate and the accuracy value with respect to the selected DNN type, i.e., $g\left(\sum_{k\in \mathcal{K}}x_{n,k}^t\theta_k \right) \left( o_n^t h_{m,c} +\left(1-o_n^t\right)h_{m,u}\right)$. Hence, the average inference accuracy for the $m$-th service in time slot $t$ can be given by
\begin{equation}\label{equ: average_accuracy}
\begin{split}
A_m\left(\mathbf{X}^t, \mathbf{o}^t\right)
&=\sum_{ n\in \mathcal{N}_m }\frac{1}{|\mathcal{N}_m |}g\left(\sum_{k\in \mathcal{K}}x_{n,k}^t\theta_k \right) \cdot \\
&\left( o_n^t h_{m,c} +\left(1-o_n^t\right)h_{m,u}\right).
\end{split}
\end{equation}
{Note that the model can be readily extended to cases when other inference methods are adopted, since the accuracy values with respect to sampling rates and DNN types are obtained via practical experiments.}

\subsection{Problem Formulation}
DNN inference services require not only minimizing service delay, but also guaranteeing their long-term accuracy requirements, which can be modeled via a CMDP. Its action, state, reward, and state transition matrix are defined as follows:
\begin{itemize}
	\item Action: The action includes the sampling rate selection, task offloading, and edge computing resource allocation decisions, i.e., $	\hat{a}^t=\{\mathbf{X}^t, \mathbf{o}^t, \mathbf{c}^t\}.$ Note that the components of the action should satisfy following constraints: (1) $ x_{n,k}^t \in \{0, 1\}$ constrains the sampling rate selection decision; (2) $ o_{n}^t \in \{0, 1\} $ requires the binary task offloading decision; and (3) $\sum_{m\in \mathcal{M}} c_m^t\leq 1$ and $0 \leq c_m^t\leq 1$ constrain a continuous computing resource allocation decision. 
	\item State:  The state includes  local computing queues backlog of industrial IoT devices $ B_n^t$,  edge computing queues backlog $Q_m^t$,  channel conditions of industrial IoT devices $H^t_n$, and the raw data size of the generated tasks at industrial IoT devices $\xi_n^t$, i.e.,
	\begin{equation}\label{equ:state_CMDP}
	\begin{split}
	\hat{s}^t=&\{ \{ B_n^t\}_{n\in \mathcal{N}},  \{Q_m^t\}_{m\in \mathcal{M}},  \{H^t_n\}_{n\in \mathcal{N}}, \{\xi_n^t\}_{n\in \mathcal{N}}\}.
	\end{split}
	\end{equation}
	The queue backlogs, i.e., $\{ B_n^t\}_{n\in \mathcal{N}}$ and $\{Q_m^t\}_{m\in \mathcal{M}}$, adopt a unit in bits, which  result in large state space, especially for a large number of industrial IoT devices.

	\item Reward: The reward is designed to minimize the service delay in \eqref{Problem: problem1} in time slot $t$, which is defined as $\hat{r}^t=-D\left(\mathbf{X}^t, \mathbf{o}^t, \mathbf{c}^t\right).$
	 	
	\item State transition probability:  State transition probability is given by
	\begin{equation}\label{equ:state_transition_CMDP}
	\begin{split}
	&\mathrm{Pr}\left(\hat{s}^{t+1}|\hat{s}^t, \hat{a}^t\right)=\prod_{n\in\mathcal{N}} \mathrm{Pr}\left(B_n^{t+1}|B_n^{t}, x_{n,k}^t, o_n^t\right)\cdot \\		&\prod_{m\in\mathcal{M}} \mathrm{Pr}\left(Q_m^{t+1}|Q_m^{t}, \mathbf{X}^t, \mathbf{o}^t\right)\cdot\\
	&\prod_{n\in\mathcal{N}} \mathrm{Pr}\left(H_n^{t+1}|H_n^{t}\right) \cdot \prod_{n\in\mathcal{N}} \mathrm{Pr}\left(\xi_n^{t+1}|\xi_n^{t}\right). 
	\end{split}
	\end{equation}
	The equality holds due to the independence of different state components. The first two components are governed by the evolution of local computing queues and edge computing queues in \eqref{equ:local_buffer} and \eqref{equ:computing_queue_update}, respectively. The third component is evolved according to the discrete-time  Markov chain of channel conditions, and the last component is governed by the memoryless task arrival pattern. Note that each of those state components only depends on its previous state components, which means the state transition is Markovian.
\end{itemize}

Our goal is to find a stationary policy $\pi\in \Pi$ that dynamically configures sampling rates selection $\mathbf{X}^t$, task offloading~$\mathbf{o}^t$, and edge computing resource allocation $\mathbf{c}^t$ according to state $\hat{s}^t$, to minimize the service delay while guaranteeing long-term inference accuracy requirements $\{A_m^{th}\}_{m\in \mathcal{M}}$, which is formulated as the following problem:
\begin{subequations}\label{Problem 1}
	\begin{align}
		{\mathbf{P}_0:}~	& \underset{\pi \in \mathrm{\Pi}}{\text{min}} ~~ \lim\limits_{T\to \infty} \frac{1}{T}\sum_{t=1}^{T}\mathbb{E}_{\pi}\left[D\left(\mathbf{X}^t, \mathbf{o}^t, \mathbf{c}^t\right) \right]  \\
	& \text{s.t.}~  \lim\limits_{T\to \infty}\frac{1}{T} \sum_{t=1}^{T}  A_m\left(\mathbf{X}^t, \mathbf{o}^t\right) \geq A_m^{th}, \forall m\in \mathcal{M}.   
	\end{align}
\end{subequations} 
Here, $\mathbf{P}_0$ is a CMDP. Directly solving the above CMDP via dynamic programming solutions~\cite{altman1999constrained} is challenging due to the following reasons. Firstly, state transition probability is unknown due to the lack of statistic information on the channel condition variation and task arrival patterns of all industrial IoT devices. Secondly, even the state transition probability is known, large action space and state space that grow with respect to the number of industrial IoT devices incur an extremely high computational complexity, which makes dynamic programming solutions intractable. Hence, we propose a deep RL-based algorithm to solve the CMDP, which can be applied in large-scale networks without requiring statistic information of network dynamics. 


\section{Deep RL-based Sampling Rate Adaption and Resource Allocation Algorithm}\label{sec: Proposed solution}
	As mentioned before, a CDMP problem cannot be directly solved via traditional RL algorithms. We first leverage the Lyapunov optimization technique to deal with the long-term constraints and transform the problem into an MDP. Then, we develop a deep RL-based algorithm to solve the MDP. To further reduce the training complexity, an optimization subroutine is embedded to directly obtain the optimal edge computation resource allocation.

\subsection{Lyapunov-Based Problem Transformation}
{The major challenge in solving problem $\mathbf{P}_0$ is to handle the long-term constraints. We leverage the Lyapunov technique~\cite{neely2010stochastic, luo2019adaptive} to address this challenge. The \emph{core idea} is to construct accuracy deficit queues to characterize the satisfaction status of the long-term accuracy constraints, thereby guiding the learning agent to meet the long-term accuracy constraints. The problem transformation procedure is presented as follows.}

Firstly, we construct inference accuracy \emph{deficit queues}  for all services, whose dynamics evolves as follows:
\begin{equation}\label{equ:accuracy_queue_update}
Z^{t+1}_m=\left[  A_m^{th}-A_m\left(\mathbf{X}^t, \mathbf{o}^t\right)+Z^{t}_m  \right]^+, \forall m \in \mathcal{M}.
\end{equation}
Here, $Z^{t}_m$ indicates the deviation of the achieved instantaneous accuracy from the long-term accuracy requirement, whose initial state is set to $Z^{0}_m=0$. Then, a Lyapunov function  is introduced to characterize the satisfaction status of the long-term accuracy constraint, which is defined as $L\left(Z^{t}_m\right)=\left(Z^{t}_m\right)^2/2$~\cite{neely2010stochastic, luo2019adaptive, xu2018joint}.  A smaller value of $L\left(Z^{t}_m\right)$ indicates better satisfaction of the long-term  accuracy constraint.

Secondly, the Lyapunov function should be consistently pushed to a low value in order to guarantee the long-term accuracy constraints. Hence, we introduce a \emph{one-shot Lyapunov drift} to capture the variation of the Lyapunov function across two subsequent time slots~\cite{neely2010stochastic}. Given $Z^{t}_m$, the one-shot Lyapunov drift is defined as $\Delta\left(Z^{t}_m\right) =L\left(Z^{t+1}_m\right)-L\left(Z^{t}_m\right)$, which is upper bounded by
\begin{equation}\label{equ:drift}
\begin{split}
&\Delta\left(Z^{t}_m\right) =\frac{1}{2}\left(\left(Z^{t+1}_m\right)^2-\left(Z^{t}_m\right)^2\right)\\
& \leq\frac{1}{2}\left(\left(Z^{t}_m + A_m^{th}-A_m\left(\mathbf{X}^t, \mathbf{o}^t\right)\right)^2-\left(Z^{t}_m\right)^2\right)\\
&= \frac{1}{2}\left(A_m^{th}-A_m\left(\mathbf{X}^t, \mathbf{o}^t\right)\right)^2+Z^{t}_m\left(A_m^{th}- A_m\left(\mathbf{X}^t, \mathbf{o}^t\right)\right)\\
&\leq C_m+Z^{t}_m\left(A_m^{th}-A_m\left(\mathbf{X}^t, \mathbf{o}^t\right)\right),
\end{split}
\end{equation}
where $C_m=\left(A_m^{th}-{A_m^{min} }\right)^2/2$ is a constant, and $A_m^{min}$ is the lowest inference accuracy that can be achieved for service~$m$. The first inequality is due to the substitution of \eqref{equ:accuracy_queue_update}, and the second inequality is because $A_m\left(\mathbf{X}^t, \mathbf{o}^t\right)\geq A_m^{min}$. 

Thirdly, based on the Lyapunov optimization theory, the original CMDP of minimizing the service delay while guaranteeing the long-term accuracy requirements boils down to minimizing a \emph{drift-plus-cost}, i.e.,
\begin{equation}\label{equ:drift_plus_cost}
\begin{split}
&\sum_{m\in \mathcal{M}}\Delta\left(Z^{t}_m\right) +V \cdot D\left(\mathbf{X}^t, \mathbf{o}^t, \mathbf{c}^t\right) \\
&\leq \sum_{m\in \mathcal{M}}C_m+\sum_{m\in \mathcal{M}}Z^{t}_m\left(A_m^{th}-A_m\left(\mathbf{X}^t, \mathbf{o}^t\right)\right)\\
&+V\cdot D\left(\mathbf{X}^t, \mathbf{o}^t, \mathbf{c}^t\right),
\end{split}
\end{equation}
where the inequality is due to the upper bound in \eqref{equ:drift}. Here $V$ is a positive parameter to adjust the tradeoff between the service delay minimization and the satisfaction status of the long-term accuracy constraints. The underlying rationale is that, if the long-term accuracy constraint is violated, i.e., $Z_m^t>0$, stratifying the long-term constraints by improving the instantaneous inference accuracy becomes more urgent than reducing the service delay. 

{In this way, the CMDP is transformed into an MDP with the objective of minimizing the drift-plus-cost in each time slot.}

\subsection{Equivalent MDP}
In the equivalent MDP, the action, state, reward, and state transition matrix are modified as follows due to the incorporation of accuracy deficit queues.
\begin{itemize}
	\item  Action: The action is the same as that in the CMDP, i.e., ${a}^t=\hat{a}^t=\{\mathbf{X}^t, \mathbf{o}^t, \mathbf{c}^t\}$.
	\item  State:  Compared with the state of the CMDP, the accuracy deficit queue backlog of services $\{Z_m^t\}_{m\in \mathcal{M}}$ should be incorporated, i.e.,
	\begin{equation}
	\begin{split}
	s^t=\{ \hat{s}^t, \{Z_m^t\}_{m\in \mathcal{M}}\}.
	\end{split}
	\end{equation}
	\item Reward: The reward is modified to minimize the drift-plus-cost in \eqref{equ:drift_plus_cost} in time slot $t$, i.e.,
	\begin{equation}\label{learning_reward}
	\begin{split}
		r^t&=-V\cdot D\left(\mathbf{X}^t, \mathbf{o}^t, \mathbf{c}^t\right) \\
	&- \sum_{m\in \mathcal{M}}Z^{t}_m\left(A_m^{th}-A_m\left(\mathbf{X}^t, \mathbf{o}^t\right)\right).
	\end{split}
	\end{equation}
	Note that the constant term $\sum_{m\in \mathcal{M}}C_m$  in \eqref{equ:drift_plus_cost} is ignored in the reward for brevity. 
	\item State transition probability: Since accuracy deficit queue backlogs are incorporated in the state, the state transition probability evolves according to 
	\begin{equation}
	\begin{split}
	\mathrm{Pr}\left(s^{t+1}|s^t, a^t\right)
	&=\mathrm{Pr}\left(\hat{s}^{t+1}|\hat{s}^t, \hat{a}^t\right)\cdot\\
	& \prod_{m\in\mathcal{M}} \mathrm{Pr}\left(Z_m^{t+1}|Z_m^{t}, \mathbf{X}^t, \mathbf{o}^t\right). 
	\end{split}
	\end{equation}
	where the second term is the evolution of the accuracy deficit queue backlog according to \eqref{equ:accuracy_queue_update}. Note that the overall state transition is still Markovian.
\end{itemize}
Then, problem $\mathbf{P}_0$ is transformed into the following MDP problem: 
\begin{equation}\label{Problem: problem1}
\begin{split}
	{\mathbf{P}_1:~}  \underset{\pi \in \mathrm{\Pi}}{\text{min}}  ~~\lim\limits_{T\to \infty}\frac{1}{T}\sum_{t=1}^{T} &\mathbb{E}_{\pi}\left[\sum_{m\in \mathcal{M}}Z^{t}_m\left(A_m^{th}-A_m\left(\mathbf{X}^t, \mathbf{o}^t\right)\right) \right.\\
	&\left.+V\cdot D\left(\mathbf{X}^t, \mathbf{o}^t, \mathbf{c}^t\right) \right].
\end{split}
\end{equation}
Similar to CMDP, solving an MDP via dynamic programming solutions also suffers from the curse of dimensionality due to large state space. Hence, we propose a deep RL-based algorithm to solve the MDP, which is detailed in Section~\ref{sec:proposed_algorithm}.



\subsection{Optimization Subroutine for Edge Computing Resource Allocation }
Although $\mathbf{P}_1$ can be directly solved by RL algorithms, an inherent property on edge computing resource allocation can be leveraged, in order to reduce the training complexity of RL algorithms. Through analysis on \eqref{Problem: problem1}, the edge computing resource allocation is independent of the inference accuracy performance, and hence it only impacts the one-shot service delay performance. In time slot $t$, once task offloading and sampling rate selection decisions are made, the optimal computing resource allocation decision can be obtained via solving the following optimization problem:
\begin{subequations}	\label{Problem 2}
\begin{align}
		{\mathbf{P}_2:\;}	 \underset{\mathbf{c}^t}{\text{min} }~~  &  D\left(\mathbf{X}^t, \mathbf{o}^t, \mathbf{c}^t\right) \nonumber \\
	 \text{s.t.}~&\sum_{m\in \mathcal{M}} c_m^t\leq 1\\
	& 0 \leq c_m^t\leq 1. \label{equ:P2_c2}
\end{align}
\end{subequations} 
A further analysis of \eqref{equ:overall_delay} indicates that only the task processing delay and queuing delay at the AP are impacted by  the edge computing resource allocation, i.e., $\sum_{n\in \mathcal{N}}\left(d^t_{n, p}+d^t_{n,q}+d^t_{n,w}\right)$. In addition, the aggregated delay from the perspective of all devices is equivalent to the aggregated delay from the perspective of all services. Hence, the objective function in $\mathbf{P}_2$ can be rewritten as $ \sum_{m\in \mathcal{M}}d_m^t$, where
\begin{equation}
\begin{split}
d_m^t&=\sum_{ n\in \mathcal{N}_m }\left( \frac{ \eta_{m,u} \left(1-o_n^t\right)\zeta\left(\mathbf{x}_n^t\right) }{c_m^t f_b}+\frac{ Q_{m}^t\eta_{m,u}  }{c_m^t f_b} \right.\\
&+\left. \frac{\eta_{m,u}   \sum_{i\neq n, i \in \mathcal{N}_m }\left(1-o_{i}^t\right) \zeta\left(\mathbf{x}_{i}^t\right) }{ 2 c_m^t f_b}\right)
\end{split}
\end{equation}
denotes the experienced delay of the $m$-th service. By analyzing the convexity property of the problem, we have the following theorem to obtain the optimal edge computation resource allocation in each time slot.

\begin{theorem}\label{theorem: optimal_computing_resource_allocation}
	The optimal edge computing resource allocation for problem $\mathbf{P}_2$ is given by
	\begin{equation}\label{equ:optimization_subroutine}
	c_m^{t,\star}=\frac{\sqrt{\Lambda^t_m}}{\sum_{m\in \mathcal{M}}\sqrt{\Lambda^t_m}}, \forall m \in \mathcal{M},
	\end{equation}
	where 
	\begin{equation}\label{equ:Lambda_definition}
	\begin{split}
		\Lambda_m^t=&\sum_{ n\in \mathcal{N}_m }\left({ \eta_{m,u} \left(1-o_n^t\right)\zeta\left(\mathbf{x}_n^t\right) }+{ Q_{m}^t\eta_{m,u}  }\right.\\ 
		&\left.+\frac{\eta_{m,u} }{2}  \sum_{i\neq n, i \in \mathcal{N}_m }\left(1-o_{i}^t\right) \zeta\left(\mathbf{x}_{i}^t\right) \right).
	\end{split}
	\end{equation}
\end{theorem}
\begin{proof}
	Proof is provided in Appendix \ref{appendix:theorem: optimal_computing_resource_allocation}.
\end{proof}

This optimization subroutine for the edge computing resource allocation is embedded in the following proposed deep RL-based algorithm. In this way, the training complexity can be reduced, because it is no longer necessary to train the neural networks to obtain optimal edge computing resource allocation policy.


\subsection{Deep RL-based Algorithm}\label{sec:proposed_algorithm}
To solve problem $\mathbf{P}_1$, we propose a deep RL-based algorithm, which is extended from the celebrated DDPG algorithm~\cite{lillicrap2015continuous}. The main difference between DDPG and the proposed algorithm is that the above optimization subroutine for computing resource allocation is embedded to reduce the training complexity. The proposed algorithm  can be deployed at the AP, which  collects the entire network state information and enforces the policy to all connected industrial IoT devices.

\begin{algorithm}[t]	\label{algorithm:SARA}
	\SetAlgoLined
	
	\textbf{Initialization}: Initialize all neural networks and the experience replay memory;\\
	\For{each episode}{
		Reset the environment and obtain initial state $s_0$\;
		\For{time slot $t\in \mathcal{T}$}{
			Determine the sampling rate selection and task offloading actions $\{\mathbf{X}^t, \mathbf{o}^t\}$ by the actor network according to current state $s^t$\;
			Determine edge computing resource allocation action $\mathbf{c}^t$ by  \eqref{equ:optimization_subroutine}\;
			Send joint action $a^t=\{\mathbf{X}^t, \mathbf{o}^t, \mathbf{c}^t\}$ to all industrial IoT devices by the AP\;
			Execute the joint action at industrial IoT devices\;
			Observe reward $r^t$ and new state $s^{t+1}$\;
			Store transition $\{s^t, a^t, r^t, s^{t+1}\}$ in the epxerience replay memory\;
			Sample a random minibatch transitions from the epxerience replay memory\;
			Train the critic and actor network by \eqref{equ:critic_update} and \eqref{equ:actor_update}, respectively\; 
			Update target networks by \eqref{equ:target_update}\;
		}
	}
	\caption{Deep RL-based algorithm for sampling rate adaption and resource allocation}
\end{algorithm}

In the algorithm, the learning agent has two parts: (a) an actor network, which is to determine the action based on the current  state; and (b) a critic network, which is to evaluate the determined action based on the reward feedback from the environment.  Let $\mu(s|\phi^\mu)$  and  $Q(s,a|\phi^Q)$ denote the actor network and the critic network, respectively, whose neural network weights are $\phi^\mu$ and $\phi^Q$. As shown in Algorithm \ref{algorithm:SARA}, the deep RL-based algorithm operates in a time-slotted manner, which consists of the following three steps.

The first step is to obtain experience by interacting with the environment. Based on current network state $s^t$, the actor network generates the sampling rate selection and task offloading actions with an additive policy exploration noise that follows Gaussian distribution $\mathcal{N}\left(0, \sigma^2\right)$. The optimization subroutine generates the edge computation resource allocation action. Then, the joint action is executed at all industrial IoT devices. The corresponding reward $r^t$ and the next state $s^{t+1}$ are observed from the environment. The state transition $\{s^t, a^t, r^t, s^{t+1} \}$ is stored in the experience replay memory for actor and critic network training.
	 
The second step is to train the actor and critic network based on the stored experience. To avoid the divergence issue caused by DNN, a minibatch of transitions are randomly sampled from the experience replay memory to break experience correlation. The critic network is trained by minimizing the loss function
\begin{equation}\label{equ:critic_update}
Loss\left(\phi^Q\right)=\frac{1}{N_b}\sum_{i=1}^{N_b}\left(y_i-Q(s_i,a_i|\phi^{Q})\right)^2,
\end{equation}
where $y_i=r_i+\gamma Q'(s_{i+1}, \mu'(s_{i+1}|\phi^{\mu'})|\phi^{Q'})$, and $N_b$ is the minibatch size. Here, $\mu'(s|\phi^{\mu'})$  and  $Q'(s,a|\phi^{Q'})$  represent actor and critic target networks with weights $\phi^{\mu'}$ and $\phi^{Q'}$. The actor network is trained via the policy gradient
\begin{equation}\label{equ:actor_update}
\nabla_{\phi^\mu}\approx\frac{1}{N_b}\sum_{i=1}^{N_b}\nabla_a Q(s_i,a|\phi^{Q})|_{s=s_i, a=\mu(s_i)}\nabla_{\theta^\mu}\mu(s_i|\phi^\mu)|_{s_i}.
\end{equation}

The third step is to update target networks. In order to ensure network training stability, the actor and critic target networks are softly updated by
	 \begin{equation}\label{equ:target_update}
	 \begin{split}
	 \phi^{Q'}=\delta \phi^{Q}+(1-\delta) \phi^{Q'},
	 \phi^{\mu'}=\delta \phi^{\mu}+(1-\delta) \phi^{\mu'},
	 \end{split}
	 \end{equation}
	 where $0<\delta\ll 1$ denotes the target network update ratio.

\begin{table}[t]
	\footnotesize
	\centering
	\caption{Simulation parameters~\cite{petrov2017vehicle, wu2019fast}.}
	\label{Table: Simulation parameters}
	\begin{tabular}{c|c}
		\hline
		\hline
	  \textbf{Parameter} & \textbf{Value}  \\
		\hline
	 Thermal noise spectrum density $(N_o)$& $ -174$ dBm/Hz~\cite{wu2019fast}\\
	 Communication bandwidth	$(W)$&  $[5, 25]$  MHz\\
	 Transmit power $(P_T)$ & 20\;dBm ~\cite{petrov2017vehicle}\\
	Average	task arrival rate $(\lambda)$ & $[0.6, 1]$ request/sec\\
		Noise figure $(N_f)$ & 5 dB\\
	 Intensity of compressed DNN $(\eta_{1,c}, \eta_{2,c})$ & $(80, 160)$  cycles/bit\\
	  Intensity of uncompressed DNN $(\eta_{1,u}, \eta_{2, u})$ &  (200, 400) cycles/bit\\
	Device and edge server CPU frequency 	$(f_n, f_b)$ & $(0.1, 2)$ GHz\\
	Number of Type I/II devices $(N_1, N_2)$ & $5, 5$\\
	 Time slot duration	$(\tau)$ & 1 second\\
	Balance parameter 	(V) & 0.05\\
	Unit penalty for queue overflow ($w_p$)& 1\\
   Accuracy of compressed DNN	$(h_{1,c}, h_{2,c}) $ & $0.8, 0.8$ \\
	Accuracy of uncompressed DNN	 $(h_{1,u}, h_{2,u})$ & $1, 1$ \\
	{Local/edge  queue capacity	$(B_{max}, Q_{max})$}&{ {(3.84, 19.2)\;megabits}} \\
		\hline
	\end{tabular}
			\vspace{-0.2cm}
\end{table}	

\section{Simulation Results}\label{sec: simulation results}
\subsection{Simulation Setup}
\begin{table}[t]
	\footnotesize
	\centering
	\caption{Parameters of the proposed RL-based algorithm.}
	\label{Neural network parameters}
	\begin{tabular}{ c c|c c }
		\hline
			\hline
		\textbf{Parameter} & \textbf{Value} & \textbf{Parameter} & \textbf{Value}\\ \hline
		Actor learning rate & $10^{-4}$& Critic learning rate & $10^{-3}$ \\
		Actor  hidden units& $(64, 32)$& Critic hidden units& $(64, 32)$  \\ 
		Hidden activation & ReLU& Actor output activation& Tanh\\
		Optimizer& Adam &Policy noise  $(\sigma)$ &$0.2$  \\
		Target update $(\delta)$ & $0.005$& Discount factor& $0.85$ \\ 
		Minibatch size& $64$ & Replay memory size& 100,000 \\
		Training episodes& $1,000$ & Time slots per episode& 200 \\ \hline
	\end{tabular}
		\vspace{-0.4cm}
\end{table}



{We consider a smart factory scenario in our simulation, in which industrial IoT devices, e.g., vibration sensors,  are randomly scattered. The industrial IoT devices installed on industrial facilities (e.g., robot arms) sense their operating conditions. The sensing data is processed locally or offloaded to an AP in the smart factory for processing.} The transmit power of an industrial IoT device is set to 20\;dBm~\cite{petrov2017vehicle}. The channel condition is modeled with three states, i.e., ``Good (G)'', ``Normal (N)'', and ``Bad (B)'', and the corresponding transition matrix is given by \cite{lei2016optimal}
\begin{equation}
 \mathbf{P}=
  \begin{bmatrix}
P_{GG} & P_{GN} & 0\\
 P_{NG} & P_{NN} & P_{NB}\\
 0 & P_{NB} & P_{BB}
 \end{bmatrix}
 =\begin{bmatrix}
0.3 & 0.7 & 0\\
0.25 & 0.5 & 0.25\\
0 & 0.7 & 0.3
\end{bmatrix}.	
\end{equation}   

\begin{figure}[t]
	\centering
	\renewcommand{\figurename}{Fig.}	
	\begin{subfigure}[Average delay performance]{
			\label{Fig:delay_convergency}
			\includegraphics[width=0.35\textwidth]{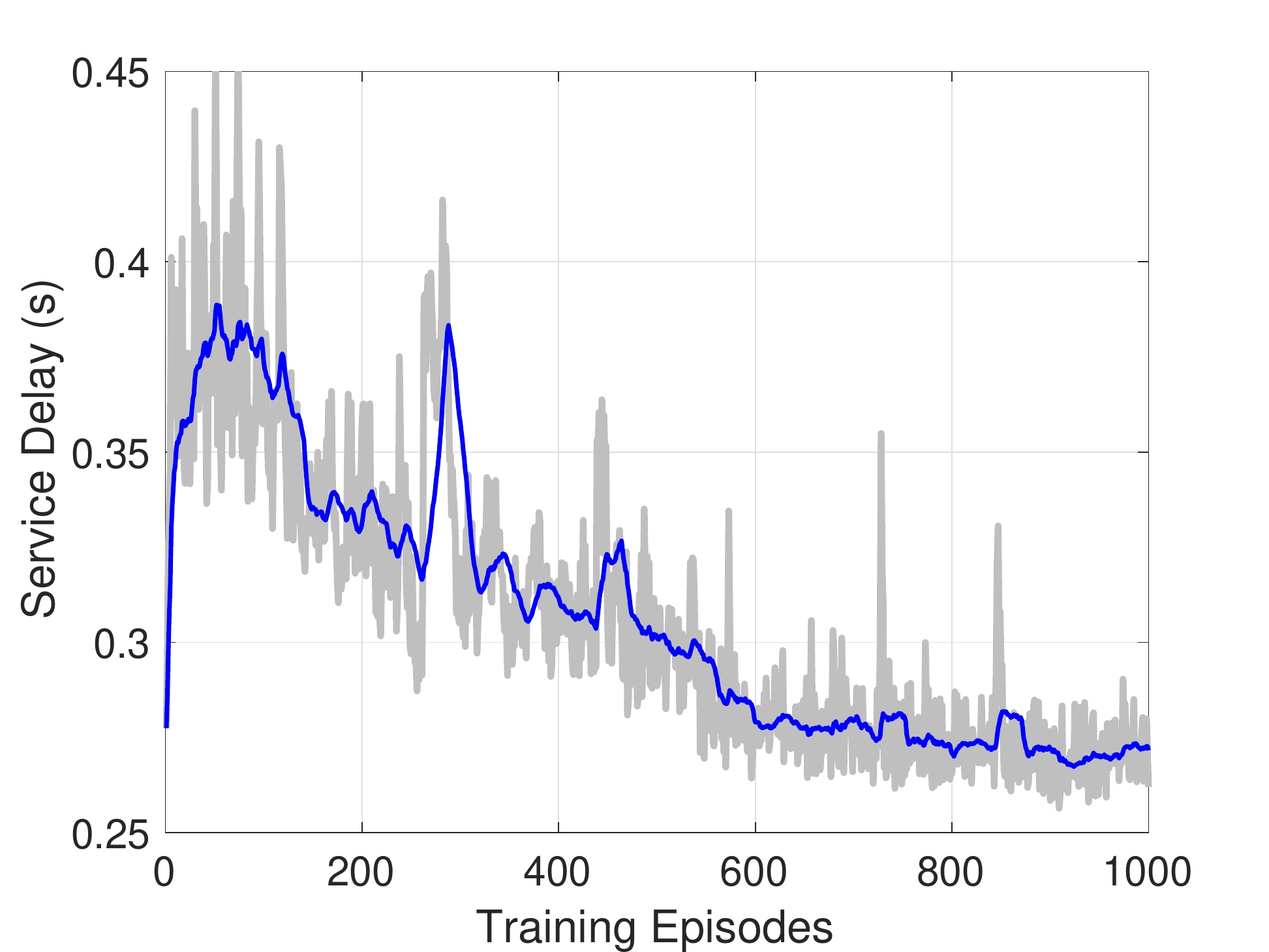}}
	\end{subfigure}
	~
	\begin{subfigure}[Accuracy Performance of two services]{
			\label{Fig:accuracyA_traffic_new.png}
			\includegraphics[width=0.35\textwidth]{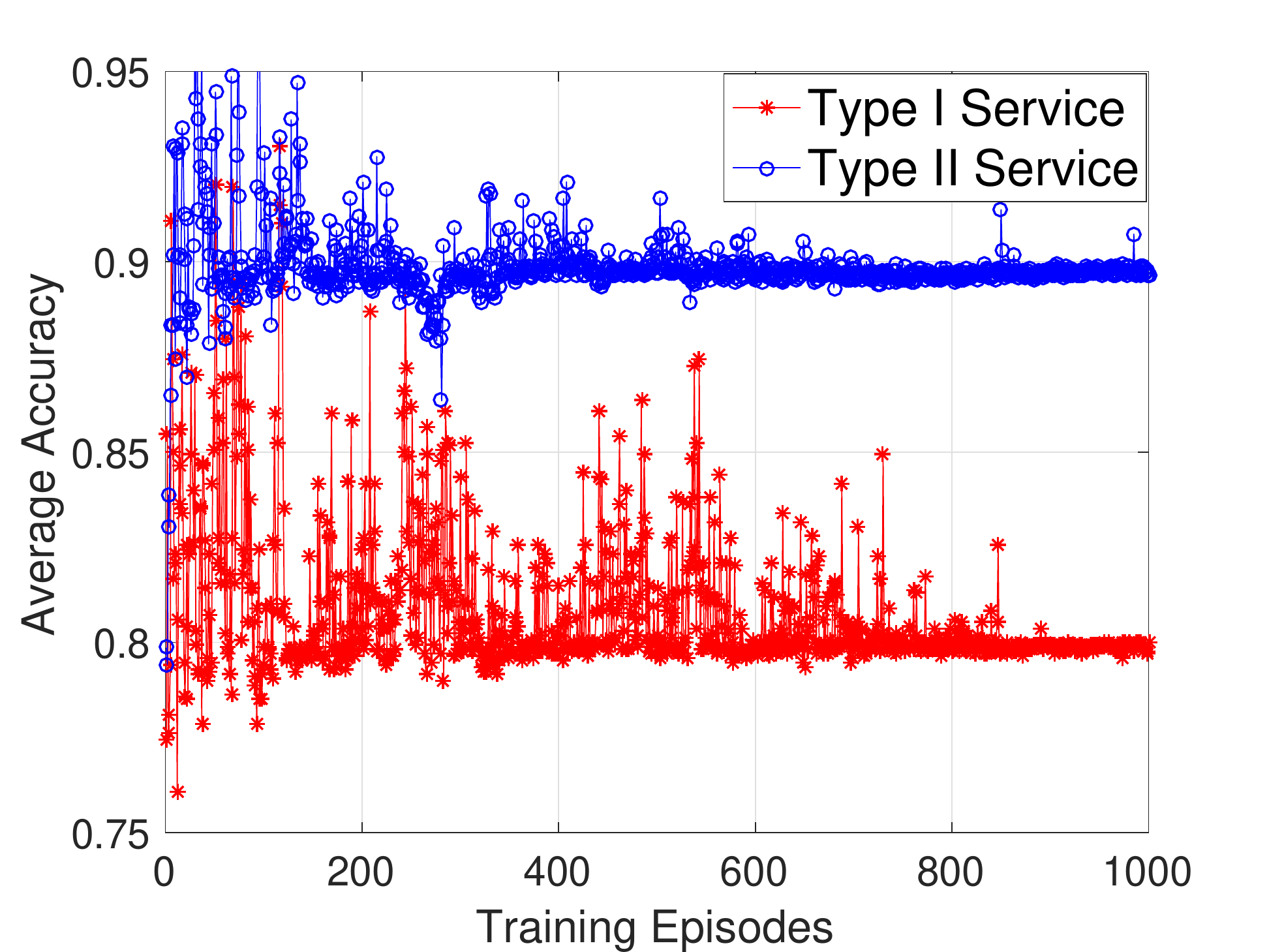}}
	\end{subfigure}
	\caption{Performance of the proposed algorithm in the training stage.}
	\label{Fig:transmission_performance}
	\vspace{-0.2cm}
\end{figure}

{Two types of DNN inference services are considered.  \emph{Type I service}: a facility fault diagnosis service to identify the fault type based on the collected  bearing vibration signal from the dataset~\cite{dataset}.  Since the duration of a time slot in the simulation is set to be one second, the task data size is the data volume of a one-second signal, which is a product of the raw sampling rate and the quantization bits of the signal. In the dataset, the bearing vibration signal is collected at 48\;KHz sampling rate  and  16\;bit quantization, and hence the corresponding task data size is 768\;kilobits. The long-term accuracy requirement of the service is set to 0.8.  \emph{Type II service}: a service extended from the Type I service to diagnose facility fault based on a low-grade bearing vibration dataset while requiring higher inference accuracy 0.9. The low-grade dataset collects vibration signal at a lower sampling rate of 32\;KHz, and hence the task data size is 512\;kilobits. For both services, the task arrival rate of each industrial IoT device at each time slot follows a uniform distribution $\mathcal{U}(\lambda -0.5, \lambda+0.5)$, where $\lambda$ is the average task arrival rate.} We consider four candidate sampling rates for industrial IoT devices, which are 25\%, 50\%, 75\% and 100\% of the raw sampling rate. The corresponding accuracy with respect to the sampling rates are 0.59, 0.884, 0.950 and 0.987, respectively, based on extensive experiments on the bearing vibration dataset~\cite{dataset}.  Balance parameter $V$ is set to 0.05 based on extensive simulations. Other important simulation parameters are listed in Table~\ref{Table: Simulation parameters}. The parameters of the proposed algorithm are given in Table~\ref{Neural network parameters}. 
The proposed algorithm  is compared to the following benchmarks: 
\begin{itemize}
	\item \textbf{Delay myopic}: Each industrial IoT device dynamically makes  sampling rate selection and task offloading decisions by maximizing the one-step reward in \eqref{learning_reward} according to the network state.
	\item \textbf{Static configuration}: Each industrial IoT device takes a static configuration on the sampling rate selection and task offloading decisions, which can guarantee services' accuracy requirements.  
\end{itemize}



\subsection{Performance Evaluation}
\subsubsection{Convergence of the proposed algorithm}
The service delay performance in the training stage is shown in Fig.~\ref{Fig:delay_convergency}. We can clearly see that the average service delay gradually decreases as the increase of training episodes, which validates the convergence of the proposed  algorithm. In addition, Fig.~\ref{Fig:accuracyA_traffic_new.png} shows the accuracy performance for both services with respect to training episodes.  {The accuracy performance is not good at the beginning of the training stage, but after 1,000 episodes of training, the accuracy performance converges to the pre-determined requirements. }


\subsubsection{Impact of task arrival rate}

\begin{figure}[t]
	\centering
	\renewcommand{\figurename}{Fig.}
	\includegraphics[width=0.35\textwidth]{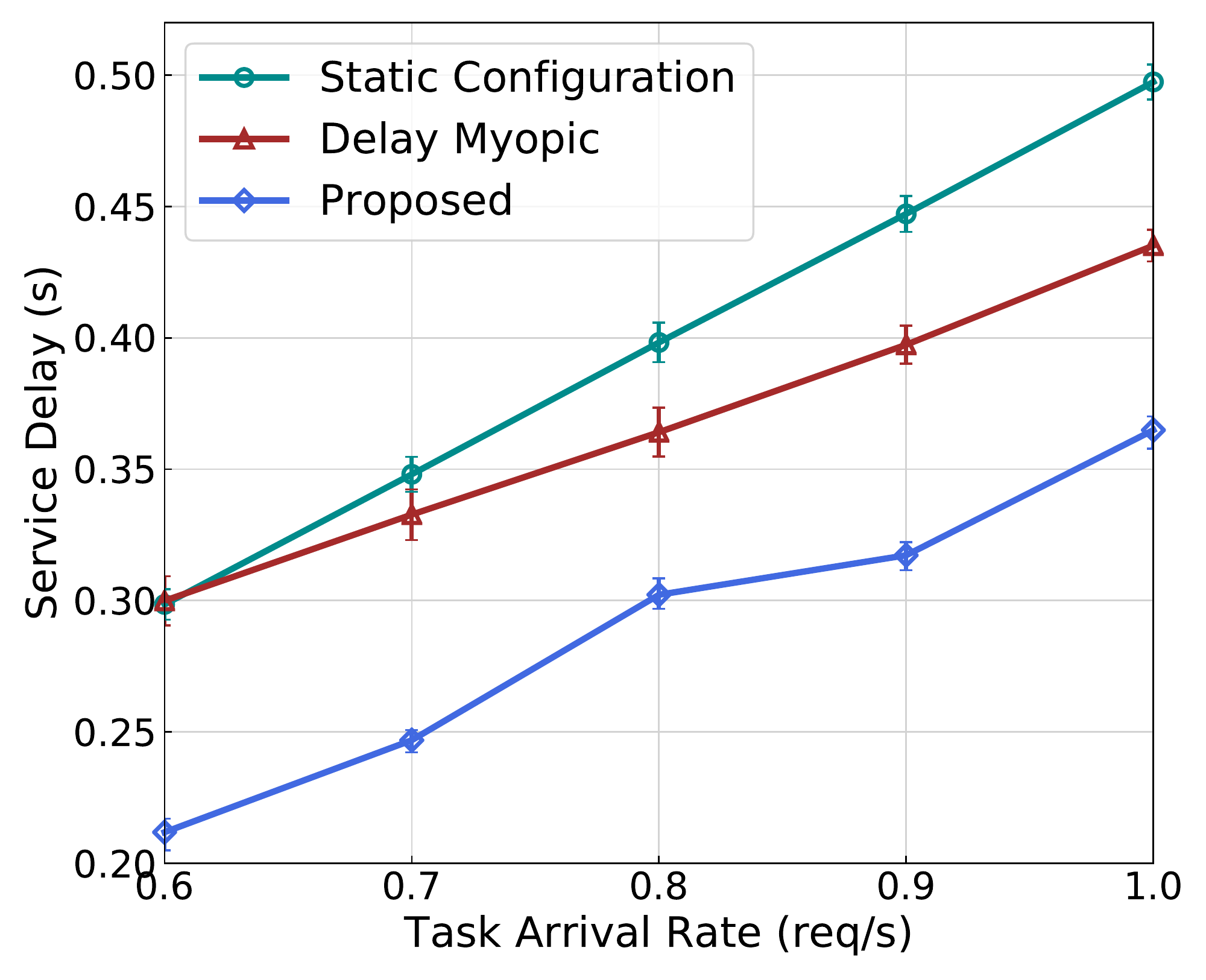}
	\caption{Service delay performance with respect to task arrival rates.}
		\vspace{-0.2cm}
	\label{Fig:delay_intensity}
\end{figure}

\begin{figure}[t]
	\centering
	\renewcommand{\figurename}{Fig.}
	\includegraphics[width=0.45 \textwidth]{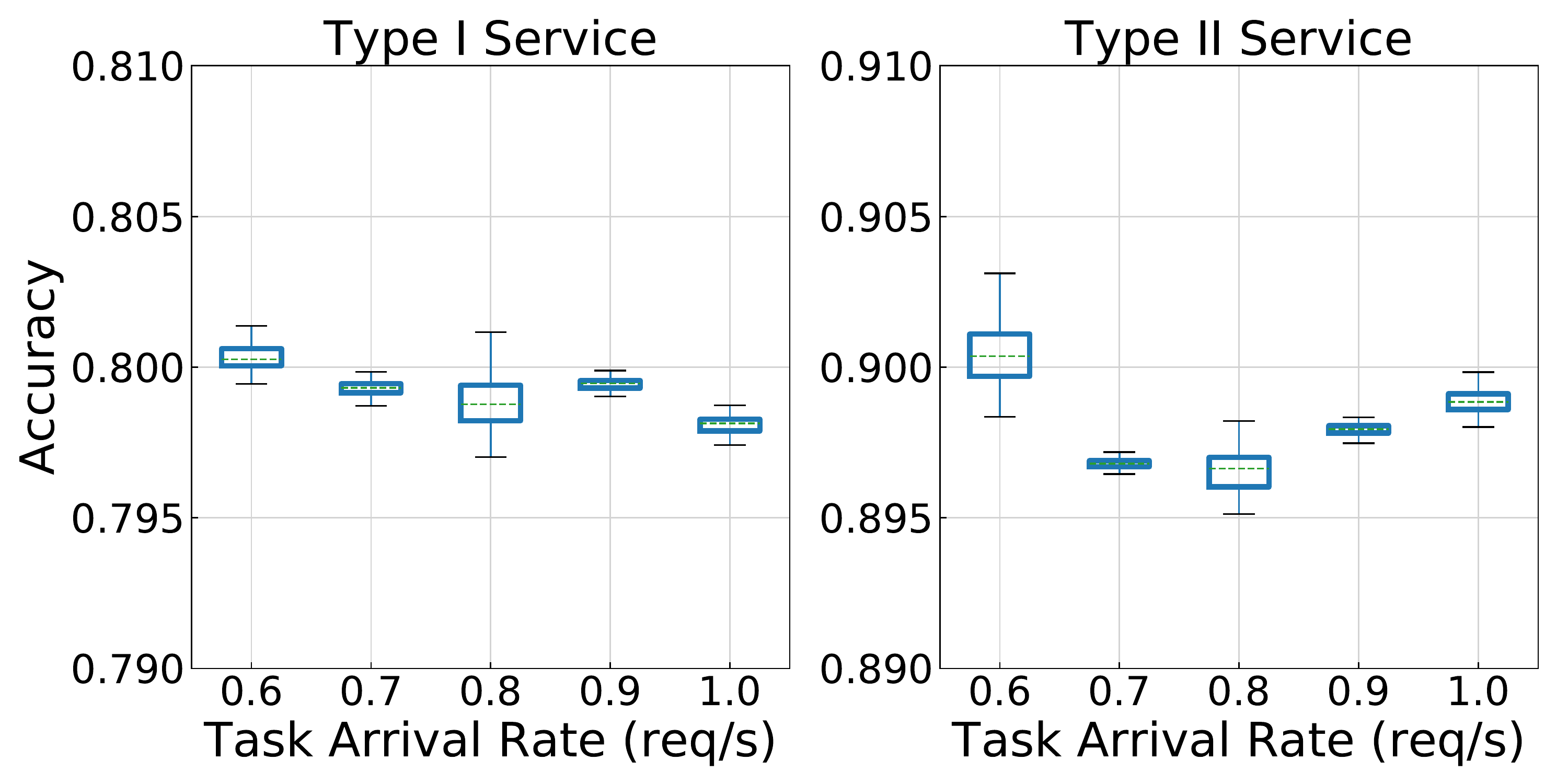}
	\caption{ Inference accuracy performance with respect to task arrival rates.}
		\vspace{-0.2cm}
	\label{Fig:Accuracy_vs_Traffic}
\end{figure}
Once well-trained offline, we evaluate the performance of the proposed algorithm in the online inference.  As shown in Fig.~\ref{Fig:delay_intensity}, we compare the average service delay performance of the proposed algorithm with benchmark schemes in terms of task arrival rates for $W=20$\;MHz. Each simulation point is plotted with a 95\% confidence interval. Several observations can be obtained from the figure. Firstly, the service delay increases with the task arrival rate due to constrained communication and computing resources in the network. Secondly, the proposed algorithm significantly outperforms benchmark schemes. {The reason is that the proposed RL-based algorithm can capture network dynamics, such as the task arrival pattern and channel condition variation, via interacting with the environment. The learned knowledge is utilized to make online decisions that target at the long-term performance, while benchmark schemes only focus on the short-term performance and do not adapt to network dynamics.} Specifically, the proposed algorithm can reduce the average service delay by 19\% and 25\%, respectively, as compared with delay myopic and  static configuration schemes. 

As shown in  Fig.~\ref{Fig:Accuracy_vs_Traffic}, boxplot accuracy distribution of two services is presented with respect to different task arrival rates. {The long-term accuracy requirements for two services are 0.8 and 0.9, respectively.} It can be seen that the proposed algorithm guarantees the long-term accuracy requirements of both services with a high probability. Specifically, the maximum error probability is less than 0.5\%.


\subsubsection{Impact of communication bandwidth}
\begin{figure}[t]
	\centering
	\renewcommand{\figurename}{Fig.}
	\includegraphics[width=0.35 \textwidth]{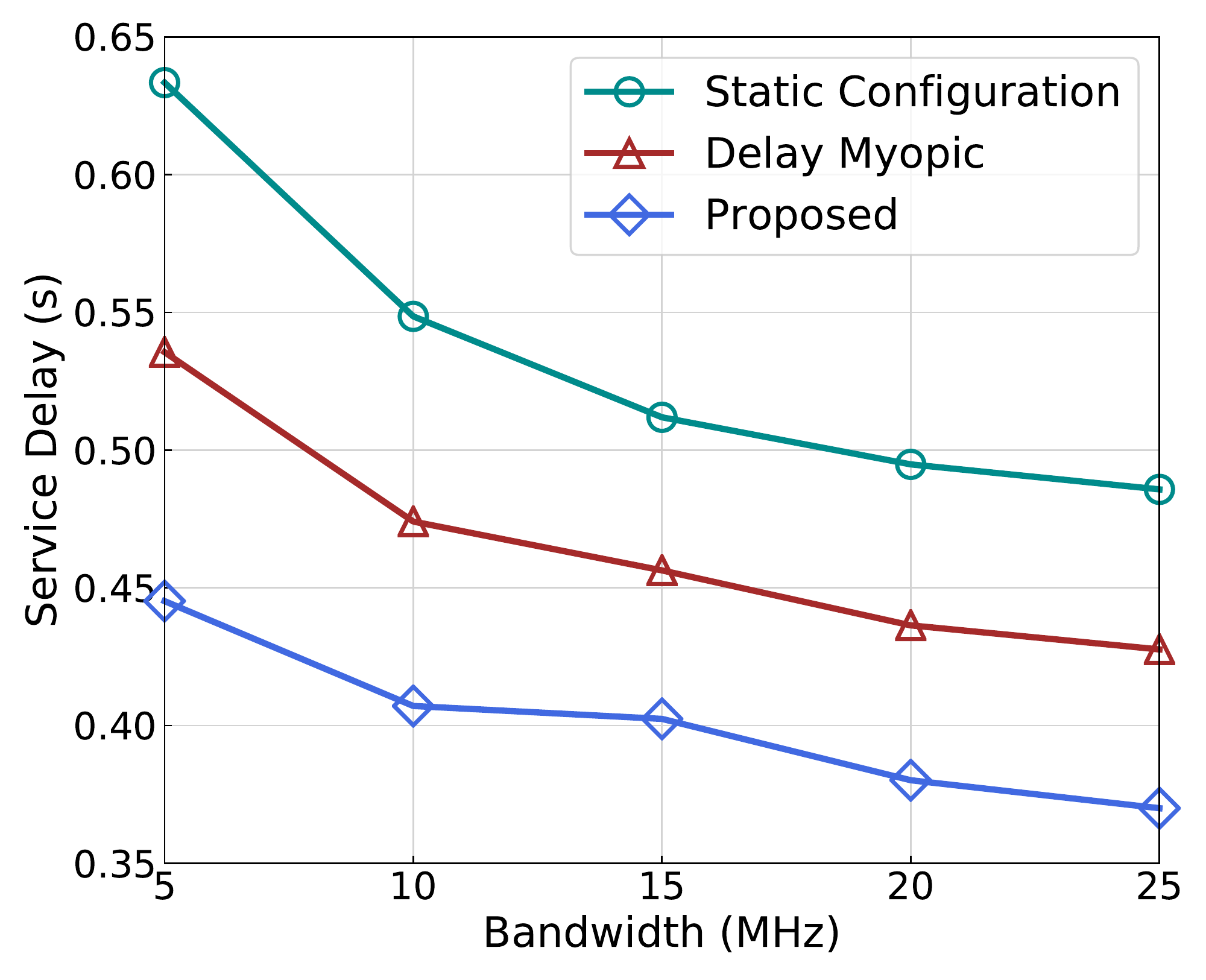}
	\caption{Service delay performance with respect to communication bandwidth.}
		\vspace{-0.2cm}
	\label{Fig:delay_bandwidth}
\end{figure}

\begin{figure}[t]
	\centering
	\renewcommand{\figurename}{Fig.}
	\includegraphics[width=0.35 \textwidth]{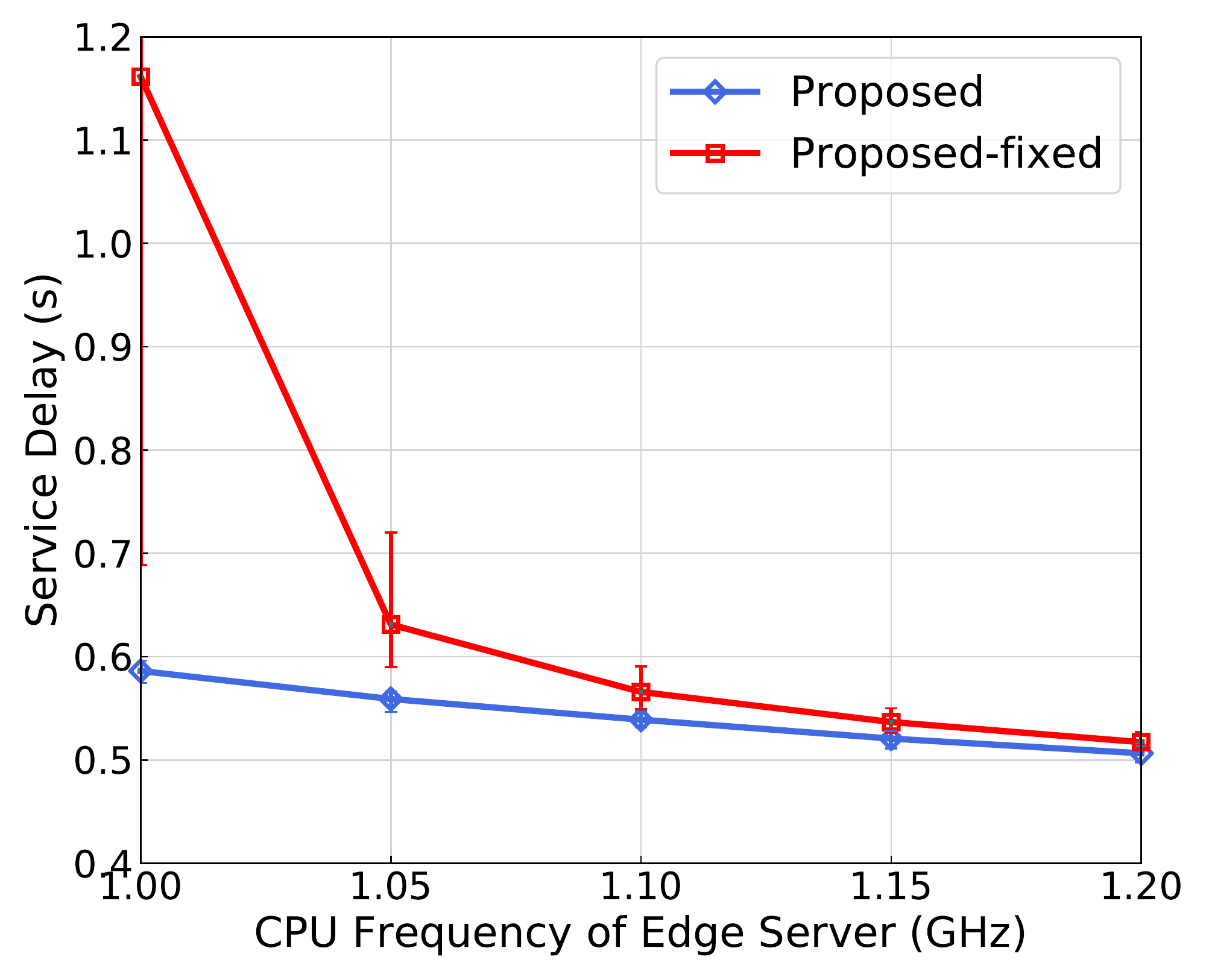}
	\caption{Service delay in terms of CPU frequency of the edge server.}
	\vspace{-0.2cm}
	\label{Fig:Delay_vs_Computing}
\end{figure}

 Fig. \ref{Fig:delay_bandwidth} shows the impact of communication bandwidth on the average service delay. Firstly, we can see that the average service delay decreases as the growth of bandwidth. The reason is that the transmission delay is reduced when the communication resource becomes sufficient. In addition, the proposed algorithm achieves good performance when the bandwidth is scarce. When system bandwidth is only 5\;MHz, the proposed algorithm achieves 1.20$\times$ and 1.42$\times$ delay reduction  compared with delay myopic and static configuration schemes, respectively, which is larger than that when the system bandwidth is 25\;MHz (1.15$\times$ and 1.31$\times$).  The reason is that the proposed algorithm efficiently utilizes the on-board computing resources. Simulation results show that the proposed algorithm decides 47.5\% computation tasks to be executed locally with 5\;MHz bandwidth, while the delay myopic benchmark only decides 17\%. Due to the efficient resource orchestration among industrial IoT devices and the AP, the proposed algorithm can  effectively reduce average service delay for both services. 

\subsubsection{Impact of  optimization subroutine}
{As shown in Fig.~\ref{Fig:Delay_vs_Computing}, we evaluate the performance of the proposed algorithm with the fixed computing resource allocation (referred to as proposed-fixed), in which the edge computing resource is allocated based on the average computing demand of two services. Compared with the proposed-fixed solution, the proposed algorithm achieves significant performance gain when the edge computing resource is constrained. Specifically, the performance gain in reducing the service delay decreases from 1.98$\times$ at 1\;GHz CPU frequency to only 1.02$\times$   at 1.2\;GHz CPU frequency. The reason is that efficient resource allocation is  more important in resource-constrained scenarios, as compared to resource-rich scenarios. The results validate the effectiveness of the optimization subroutine for edge computing resource allocation. In addition to the performance gain, another merit of the optimization subroutine is to reduce the training complexity of RL algorithms.}

\section{Conclusion}\label{sec: conclusion}
In this paper, we have studied the sampling rate adaption and resource allocation problem for collaborative DNN inference in industrial IoT networks. A deep RL-based algorithm has been developed to determine the channel variation and the task arrival pattern which are then exploited to provide accuracy-guaranteed DNN inference services. 
The proposed algorithm can optimize service delay performance on the fly, without requiring statistic information of network dynamics. The Lyapunov-based transformation technique can be applied to other CMDPs. For the future work, we will investigate the impact of device mobility on the inference performance.

\appendix
\subsection{Proof of Theorem \ref{theorem: optimal_computing_resource_allocation}}\label{appendix:theorem: optimal_computing_resource_allocation}

Firstly, the problem is  proved to be a convex optimization problem. For brevity of notations, we omit $t$ in the proof. With the definition of $\Lambda_m$ in \eqref{equ:Lambda_definition}, the objective function can be rewritten as $\sum_{m\in \mathcal{M}}{\Lambda_m}/{\left( c_m f_b\right)}$. The second-order derivative of the objective function shows   $2\Lambda_m/\left(f_b c_m^3\right) >0$. In addition, the inequality constraint is linear.  Hence, the problem is a convex optimization problem.

Secondly, a Lagrange function for the problem without considering the inequality constraints is constructed,  i.e.,
\begin{equation}
\mathcal{L}\left(\mathbf{c}, {a} \right)=\sum_{m\in \mathcal{M}}\frac{\Lambda_m}{ c_m f_b}+a\left(\sum_{m\in \mathcal{M}}c_m-1 \right),
\end{equation}  
where $a$ denotes the Lagrange multiplier. Based on Karush-Kuhn-Tucker conditions~\cite{boyd2004convex}, we have 
\begin{equation}
\begin{split}
\frac{\partial L\left(\mathbf{c}, {a} \right) }{\partial c_m}=-\frac{\Lambda_m}{f_b c_m^2}+a=0, \forall m\in  \mathcal{M}.
\end{split}
\end{equation}
By solving the above equation, we can obtain $c_m^\star=\sqrt{{\Lambda_m}/{a f_b}},  \forall m\in  \mathcal{M}.$
Substituting the above result into the complementary slackness condition $\sum_{m\in \mathcal{M}}c_m^\star-1=0$, the optimal value of $a$ is given by $a^\star ={\left(\sum_{m\in \mathcal{M}}\sqrt{\Lambda_m} \right)^2}/{f_b}$. From the above equation, $a^\star$ takes a positive value, and hence $\{c_m^\star\}_{ m\in \mathcal{M}}$ are positive values, which means constraint \eqref{equ:P2_c2}, i.e., $ c_m^t\geq 0, \forall m\in \mathcal{M}$, is automatically satisfied. Substituting $a^\star$ into the complementary slackness condition proves Theorem~\ref{theorem: optimal_computing_resource_allocation}.

\bibliographystyle{IEEEtran}
\bibliography{security}

\end{document}